\documentclass[article, aps,prd,floatfix, nofootinbib]{revtex4-2}
\usepackage{graphicx}
\usepackage{amsfonts}
\usepackage{amsmath}
\usepackage{rotating}
\usepackage{amssymb,amsthm}
\usepackage{xcolor}
\usepackage[normalem]{ulem}
\usepackage{appendix}
\usepackage{hyperref}
\usepackage{longtable}
\usepackage{footnote}

\newtheorem{thm}{Theorem}[section]

\newtheorem{prop}[thm]{Proposition}

\theoremstyle{definition}
\newtheorem{defn}[thm]{Definition}
\theoremstyle{remark}
\newtheorem{rem}[thm]{Remark}

\newcommand{\mbold}[1]{\mbox{\boldmath{\ensuremath{#1}}}}

\def\beq{\begin{eqnarray}}
\def\eeq{\end{eqnarray}}

\def \bell {\mbox{{\mbold\ell}}}
\def \bn {\mbox{{\bf n}}}
\def \bm {\mbox{{\bf m}}}
\def \bh {\mbox{{\bf h}}}
\def \bx {\mbox{{\bf x}}}

\def \bomega {\mbox{{\mbold \omega}}}

\begin{document}


\title{All symmetry groups of pp-waves in teleparallel gravity}

\author{A. Landry}
\email{a.landry@dal.ca}
\affiliation{Department of Mathematics and Statistics, Dalhousie University, Halifax, Nova Scotia, Canada, B3H 3J5}

\author {D. D. McNutt}
\email{mcnuttdd@gmail.com}
\affiliation{The Royal Norwegian Naval Academy, N-5165 Bergen, Norway}



\begin{abstract}

We elaborate on and further develop an approach to determining the teleparallel analogue of spacetimes in General Relativity (GR) by studying the Teleparallel analogue of  pp-Wave (TppW) spacetimes. This relies on using the fact that these solutions belong to the Vanishing Scalar Invariant (VSI) subclass for which the explicit forms of the frame and spin-connection are known. By identifying the pp-wave (ppW) metric within this class, we are able to use frame based symmetry methods and the Cartan-Karlhede (CK) algorithm to determine the necessary form for the frame. Through this analysis we find two overlooked solutions that are permitted in teleparallel gravity (TPG) and in GR. 

\end{abstract}

\maketitle


\section{Introduction}\label{sec:Intro}

Gravity theories describe not only of spacetime structures like Black Hole (BH) and galaxy formation, but wave processes as well. One only has to think of gravitational waves (GW) whose definition for plane waves can be summarized as a perturbation of Minkowski spacetime \cite{Weinberg,Einstein:1937qu,Flanagan:2005yc}. GWs coming from a binary system of two merging BHs were observed for the first time during the LIGO and VIRGO experiments in 2016 \cite{LIGOScientific:2016aoc,LIGOScientific:2016nwa,LIGOScientific:2016wyt,Arnaud:2016wen,LIGOScientific:2016kum}. Meanwhile, there have also been other theories concerning such GWs and their possible generalizations \cite{Flanagan:2005yc,Coller:1958tx}. There are also other types of the same wave phenomena type such as matter waves \cite{Lanczos:1942zz,Balazs:1958zz,Peres:1960wfo,Penrose:1965rx}. There has also been a multitude of works on gravitational thermal radiation such as Hawking radiation for spherically symmetric objects (see refs \cite{Hawking:1974rv,Hawking:1975vcx,Wald:1975kc,Arefeva:1995xt} for details) as well as Unruh radiation coming from linearly accelerated objects. In the latter case, there is the emergence of oscillatory processes \cite{Hammad:2021rey,Alsing:2004ig,Ford:2005sh,Akhmedov:2007xu}. This is particularly the case of electromagnetic waves from Unruh radiation of these phenomena one of the most credible toy model of physical wave processes are the parallel propagating waves (pp-Waves or plane-fronted waves) \cite{Ehlers:1962zz,Lee:2005ag,Gurses:1978zs,Aichelburg:1970jk,kramer,sippel1986symmetry}.The existence of GWs and their effects give insight into the validity of alternative gravity theories as well. In teleparallel gravity, GWs have been studied indirectly but without full knowledge of the solutions form.  \cite{hohmann2018propagation, emtsova2024equivalence}. These studies have assumed the existence of plane gravitational waves which are a special case of the pp-wave spacetimes admitting a high degree of symmetry.

The pp-Waves (ppW) are characterized by waves propagating in a timelike plane with a null wave vector, $\bell$, for which the wave fronts are spatial planes which are parallel to each other. This characterization is equivalent to the wave vector $\bell$ being covariantly constant, so that $\nabla \bell = 0$. The timelike plane can be coordinatized using null coordinates $(u,v)$ so that $\bell = \partial_v$, while the spatial plane is coordinatized with spatial coordinates $(x,y)$. In this coordinate system the ppW metric is now \cite{sippel1986symmetry,kramer}:
\begin{align}\label{ppwavedefinition}
	ds^2 =& -2H(u,x,y)\,du^2 -2dudv + dx^2 +dy^2,
	\nonumber\\
	=& -2H(u,r,\theta)\,du^2 -2dudv + dr^2 +r^2\,d\theta^2,
\end{align}
where $H(u,x,y)$ is an arbitrary function. Since ppWs are described as a 4-dimensional (4D) spacetime, the latter will respect the principles of General Relativity (GR) and therefore the metric associated with ppWs will basically satisfy Einstein's equations \cite{sippel1986symmetry,kramer,Ehlers:1962zz}:
\begin{equation} \label{einsteineqns}
	\overset{\ \circ}{G}_{ab}=\kappa\,\Theta_{(ab)} ,
\end{equation}
where $\overset{\ \circ}{G}_{ab}$ is the Einstein Tensor, $\kappa$ and $\Theta_{(ab)}$ is the energy-momentum (EM). According to GR, the Ricci curvature $\overset{\ \circ}{R}_{ab}$ for these ppWs is defined by the only following non-zero component \cite{sippel1986symmetry}:
\begin{align}\label{riccicurvature}
	\overset{\ \circ}{R}_{11}=\left(H_{,22}+H_{,33}\right) \equiv \nabla^2\,H = \left( \partial_x^2+\partial_y^2\right)\,H.
\end{align}   
For spacetimes and gravity theories where the Ricci curvature is zero, the equation (eqn) \eqref{riccicurvature} becomes simply a 2-dimensional (2D) Laplace equation defined as $\nabla^2\,H=0$. In this case, any 2D Laplace equation solution for $H$ will lead to vacuum ppWs. We will see  later that this is the case for Teleparallel gravity (TPG) theories as well.

We should note that the same ppWs can also be solution for a spacetime with a null fluid and/or an electromagnetic source satisfying the Einstein-Maxwell field equations (FEs). In the latter case, the eqn \eqref{einsteineqns} becomes \cite{sippel1986symmetry}:
\begin{equation} \label{einsteinmaxwell}
	\overset{\ \circ}{G}_{ab}=\kappa\, \Theta_{(ab)}=\kappa\,\left(F_{ac}F^c_{~b}+\tfrac{1}{2}g_{ab}F_{cd}F^{cd}\right),
\end{equation}
where $F_{ac}$ is the electromagnetic stress-energy tensor. In addition to the eqn \eqref{einsteinmaxwell}, the electromagnetic ppWs will have to satisfy the conservation laws as $F^{ab}_{~;b}=0$,  $*F^{ab}_{~;b}=0$ and $dF=0$ \cite{sippel1986symmetry}. 

While the ppWs are defined by the existence of a covariantly constant null vector-field $\bell$ with $\nabla \bell = 0$, there are many subclasses of ppW depending on the choice of $H(u,x,y)$ in eqn \eqref{ppwavedefinition}. In particular, there are many different symmetry groups and their associated Lie algebras permitted within the ppW solutions \cite{sippel1986symmetry}. In GR, there are several classes of solutions ranging from 2-dimensional to 7-dimensional Lie algebras leading to different types of $H(u,x,y)$ functions and symmetry operators \cite{sippel1986symmetry,kramer}. In this paper we are interested in what symmetry groups are permitted for the teleparallel ppWs.

Before going further with the ppWs solutions and symmetries, we will briefly introduce TPG here, and discuss it in greater detail in the next section. Teleparallel $F(T)$-type gravity ($F(T)$ TPG) is an alternative theory to GR whose physical quantities are described by only the spacetime torsion tensor \cite{Aldrovandi_Pereira2013,Bahamonde:2021gfp,Krssak:2018ywd}. This physical quantity is fundamentally a function of the frame ${\bf h}_a$ (and its derivatives) as well as the spin-connection $\omega^a_{~bc}$, which gives a frame based theory instead of a metric-based theory. Geometrically speaking, a Riemannian geometry in GR is defined in terms of the curvature of the Levi-Civita (LC) connection which is ultimately defined by the metric. However, this changes radically for $F(T)$ TPG, where the curvature-free condition allows for many choices of connection.

The simplest form of teleparallel theory leading to GR-like FEs is the Teleparallel Equivalent to General Relativity (TEGR). It is based on the torsion scalar $T$ itself constructed from the torsion tensor \cite{Aldrovandi_Pereira2013}. The most common generalization is the $F(T)$ TPG, with a function $F$ of the torsion scalar $T$ \cite{Ferraro:2006jd,Ferraro:2008ey,Linder:2010py}. Covariantly, the geometries associated with the $F(T)$ TPG are gauge invariant with vanishing curvature, ${\bf R}=0$, and vanishing non-metricity ${\bf \nabla g} = {\bf Q}=0$. This requirement will permit the spin-connections $\omega^a_{~bc}=0$ in a special class of frame, known as a {\it  proper frame}, and non-vanishing for other frames \cite{Lucas_Obukhov_Pereira2009,Aldrovandi_Pereira2013,Krssak:2018ywd}. Even though TPG is locally invariant under the Lorentz definition, a proper frame is not defined in a spin-connection invariant manner and some problems may arise concerning symmetries and potential extra Degree of Freedom (DOF) problems \cite{Krssak_Pereira2015}. In addition, the above considerations are usually valid for other similar approaches like the New General Relativity (NGR) (see refs. \cite{Hayashi:1979qx,BeltranJimenez:2019nns,Bahamonde:2024zkb} and refs. within), the symmetric teleparallel $F(Q)$-type gravity (see refs. \cite{Heisenberg:2023lru,Heisenberg:2023wgk,Flathmann:2021itp,Hohmann:2021ast} and refs. within), some intermediate theories such as $F(T,Q)$-type, $F(R,Q)$-type, $F(R,T)$-type (see refs. \cite{BeltranJimenez:2019esp,Nakayama:2022qbs,Xu:2019sbp,Maurya:2024mnb,Maurya:2024ign,Maurya:2024nxx,Harko:2011kv} and refs. within).

For a given gravity theory it is necessary to determine solutions that reflect our physical Universe. In $F(T)$ TPG, the freedom of choice in choosing the connection complicates generating physically meaningful solutions. Furthermore, the teleparallel analogue of the Ricci tensor in the FEs is no longer necessarily symmetric. In the following subsection we will see that depending on the $F(T)$ TPG theory, the requirement that this tensor is related to a symmetric energy momentum tensor will impose additional conditions that must satisfied. Instead of directly solving the FEs, we will employ symmetries and symmetry methods to analyze the teleparallel analogues of the ppW spacetimes. Symmetries and symmetry methods are powerful tools in the analysis of complex mathematical and physical problems. Techniques based on symmetry have been used in GR to produce meaningful solutions with closed form expressions.

In an analoguous manner, frame based symmetry approaches have been implemented for Riemann-Cartan geometry \cite{mcnutt2024symmetries, Coley:2019zld, McNutt:2023nxm} and these can be applied to the special case of TPG. Generally, these symmetries associated with Lie groups and algebras are expressed in terms of operators called Killing Vector fields (KVFs) in GR. However, in frame based theories, not all KVFs will lead to symmetries for the geometry since any KVF, $X$ must also satisfy:
\beq 
\mathcal{L}_{{\bf X}} T^a_{~bc} = 0  \quad\text{and} \quad \mathcal{L}_{{\bf X}} \nabla^p\,T^a_{~bc} = 0, \quad\text{where}\quad p \in \mathbb{N} . \label{Intro3}
\eeq
\noindent A KVF which satisfies this constraint is an {\it  affine frame symmetry generator}. In subsection \ref{sec:Sym} we will discuss how these equations can be written as a first order system which generalizes the Killing equations for KVFs.

Our aim is to use symmetry methods to determine the TppW geometries. However, in order to do so, we require an initial frame anzatz to work with, since $\bell = \partial_v$ with $\nabla \bell = 0$ is far too general as an initial condition to determine the frame. It is well-known that the ppW solutions in GR belong to the subclass of Vanishing Scalar Invariants (VSI) spacetimes where all scalar polynomial curvature invariants vanish \cite{pravda2002all}. It was shown in  ref.~\cite{McNutt:2021tun} that the teleparallel analogue of the class of VSI metrics which contain the ppW solutions are also VSI with respect to scalar polynomial invariants (SPIs) constructed from the torsion tensor and its covariant derivatives. Furthermore in ref.~\cite{McNutt:2021tun}, the explicit proper frames were given for the teleparallel analogue of the VSI class. 

Denoting the set of all torsional SPIs as $\mathcal{I}_T$, we note that the VSI spacetimes are a subclass of the Constant Scalar Invariant (CSI) class of spacetimes where $\mathcal{I}_T$ consists of constants.  Ultimately, both classes are contained within the $\mathcal{I}_{T}$-degenerate class of teleparallel geometries where the invariants in $\mathcal{I}_T$ cannot be used to classify the teleparallel geometries in this class. It was shown in \cite{McNutt:2021tun} that these geometries can be described by proper frames constructed from the Kundt frame \cite{CHPP2009,mcnutt2013vacuum,4dCSI,Higher,McNutt:2021tun}:
\beq \bh^a = \left[ \begin{array}{c} \bn \\ \bell \\ \bar{\bm} \\ \bm \end{array} \right] = \left[ \begin{array}{c}   dv + H du + W_1 dx + W_2 dy \\ du \\ {P}(dx-idy)  \\ {P}(dx+idy) \end{array} \right], \label{AKundtframe}\eeq
\noindent after applying Lorentz transformations that preserve the $\bell$ direction. Here $H(u,v,x,y)$, $W_1(u,v,x,y)$ and $W_2(u,v,x,y)$ are arbitrary real-valued ($\mathbb{R}$-valued) functions and $P=P(u,x,y)$ is an arbitrary complex-valued ($\mathbb{C}$-valued) function. In order to be $\mathcal{I}_T$-degenerate, $H$ must be at least quadratic in $v$ while $W_1$ and $W_2$ are linear in $v$ \cite{McNutt:2021tun}. There are further conditions on the $v$-coefficients of $H$, $W_1$ and $W_2$ but these are not necessary for the current topic.  In the next section, we will provide conditions on the frame functions $H, W_1, W_2$ and $P$ required to describe a VSI TppW spacetime. 

The frame in eqn \eqref{AKundtframe} is defined in Complex Null gauge relative to this frame, the metric is \cite{Coley:2019zld}:
\begin{align}
	g_{ab} = \left[ \begin{array}{cccc} 0 & -1 & 0 & 0 \\ -1 & 0 & 0 & 0 \\ 0 & 0 & 0 & 1 \\ 0 & 0 & 1 & 0 \end{array}\right].
\end{align}
\noindent This is the most natural gauge to use for spacetimes that admit a geometrically preferred null direction, such as $\bell$ in the ppW spacetimes. 

The VSI class implies the spacetime will have zero torsion scalar, $T=0$ as this is a torsional SPI which must vanish. Due to this the FEs will be much simpler in this case since the TppW solutions must now satisfy the TEGR FEs \cite{Landry:2023egz} instead of the full $F(T)$ TPG FEs. The TppW solutions are thus a new additional class of null torsion scalar spacetime solutions, as some flat cosmological spacetimes without any spatial curvature \cite{Landry:2023egz}.

The outline of the remainder of the paper is as follows. In  section \ref{sec:teleparallelsumm}, we discuss the necessary concepts from TPG necessary to introduce the TppW spacetimes. In particular, we review the irreducible decomposition of the torsion tensor and introduce the torsion scalar $T$. After this, we review the $F(T)$ TPG FEs and discuss the implication of $T=0$. The section concludes with a review of the frame based symmetry formalism and the CK algorithm. In section \ref{sec:TppWDef} we explicitly determine the proper frame for a general TppW spacetimes without any further assumptions of symmetry. In section \ref{sec:TppWsmall}, using the fact that the symmetry group of the TppW spacetimes is a subgroup of the ppW spacetimes in GR, we use the results of ref.~\cite{sippel1986symmetry} as a framework to apply the frame based symmetry formalism and determine all permitted symmetry groups of the TppW spacetimes with dimension less than three with a trivial isotropy group. In section \ref{sec:isotropy} we determine the subclasses of TppW spacetimes that admit non-trivial linear isotropy and consequently use the CK algorithm to identify those special subcases which admit the largest symmetry group possible with a non-trivial isotropy group. In section \ref{sec:TppWLarge}  we determine all higher dimensional subgroups of the ppW isometry groups with non-trivial isotropy, which act as higher-dimensional symmetry groups for the TppW spacetimes. The results are summarized in a table in section \ref{sec: Discussion} and the results are discussed.


\section{Summary of Teleparallel Gravity}\label{sec:teleparallelsumm}

\subsection{Irreducible parts of the torsion tensor}

In TPG, all physical quantities are expressed in terms of torsion spacetime. The torsion and curvature tensors are defined as \cite{Obukhov_Pereira2003,Aldrovandi_Pereira2013,Krssak:2018ywd}:
\begin{subequations}
	\begin{eqnarray}
		T^a_{\phantom{a}\mu\nu}&=&\partial_\mu h^a_{\phantom{a}\nu}-\partial_\nu h^a_{\phantom{a}\mu}+\omega^a_{\phantom{a}b\mu}h^b_{\phantom{a}\nu}-\omega^a_{\phantom{a}b\nu}h^b_{\phantom{a}\mu}, \label{torsiondef}
		\\
		R^a_{\phantom{a}b\mu\nu} &=& \partial_\mu \omega^a_{\phantom{a}b\nu}-\partial_\nu \omega^a_{\phantom{a}b\mu}+\omega^a_{\phantom{a}c\mu}\omega^c_{\phantom{a}b\nu}-\omega^a_{\phantom{a}c\nu}\omega^c_{\phantom{a}b\mu}- c^e_{~cd} \omega^a_{~be}= 0, \label{curvature}
	\end{eqnarray}
\noindent where $[\bh_c, \bh_d] = c^e_{~cd} \bh_e$ and $c^e_{~cd}$ are the coefficients of anholonomy for the frame $\bh_a$. 
\end{subequations}

In teleparallel spacetimes, all curvature tensor components must identically satisfy $R^a_{\phantom{a}b\mu\nu}=0$. The torsion scalar $T$ is defined in terms of eqn \eqref{torsiondef} as:
\begin{equation}
	T =
	\frac{1}{4} \; T^a_{\phantom{a}bc} \, T_a^{\phantom{a}bc} +
	\frac{1}{2} \; T^a_{\phantom{a}bc} \, T^{c b}_{\phantom{aa}a} -
	T^a_{\phantom{a}c a} \, T^{b c}_{\phantom{aa}b}.
	\label{TeleLagra}
\end{equation}
This eqn \eqref{TeleLagra} is an example of torsional SPI. However, the torsion tensor $T_{abc}$ as defined in eqn \eqref{torsiondef} can be decomposed into three irreducible parts under the local Lorentz group  and by using the differential form $T^a=\frac{1}{2}T^a_{\phantom{a}bc} \,h^b \wedge h^c$ as \cite{Hehl_McCrea_Mielke_Neeman1995,Krssak:2018ywd,Coley:2019zld,McNutt:2021tun}:
\begin{align}
	T_{abc} = \frac23 (t_{abc} - t_{acb}) - \frac13 (g_{ab} V_c - g_{ac} V_b) + \epsilon_{abcd} A^d, \label{TorsionDecomp} 
\end{align}
where $V_a$ is the Vector part, $A_a$ is the Axial part, and $t_{abc}$ the Tensor part. These three irreducible parts are defined as:
\begin{subequations}
	\begin{align}
		V_a =& T^b_{~ba}, \label{Vtor}
		\\
		A^a =& \frac{1}{6} \epsilon^{abcd}T_{bcd}, \label{Ator}
		\\
		t_{abc} =& \frac{1}{2} (T_{abc}+ T_{bac}) -\frac{1}{6} (g_{ca} V_b + g_{cb} V_a) + \frac{1}{3} g_{ab} V_c. \label{Ttor}
	\end{align}
\end{subequations}
We note that in practice it will be easier to work with $\hat{T}_{abc} = \frac23 ( t_{abc} - t_{acb})$ instead of $t_{abc}$. The eqns \eqref{Vtor} to \eqref{Ttor} for the irreducible parts not only allow to fundamentally define this quantity in TPG, but we will see later that these quantities will facilitate the characterization of all new TppW solutions. Beyond the TppW solution considerations, the NGR is defined from an action integral consisting essentially of three terms corresponding to the contributions of these same torsion tensor irreducible parts \cite{Hayashi:1979qx,BeltranJimenez:2019esp,Bahamonde:2024zkb}.

\subsection{General Teleparallel F(T)-type FEs and conservation laws}

\noindent The $F(T)$ TPG action integral is \cite{Aldrovandi_Pereira2013,Bahamonde:2021gfp,Krssak:2018ywd,Coley:2019zld,Coley:2024tqe,Landry:2024dzq,Landry:2024pxm}:
\begin{equation}\label{1000}
	S_{F(T)} = \int\,d^4\,x\,\left[\frac{h}{2\kappa}\,F(T)+\mathcal{L}_{Source}\right],
\end{equation}
where $h$ is the coframe determinant and $\kappa$ is the coupling constant. By applying the least-action principle to eqn \eqref{1000}, the symmetric and antisymmetric parts of FEs are \cite{Coley:2024tqe,Landry:2024dzq,Landry:2024pxm}:
\begin{subequations}
	\begin{eqnarray}
		\kappa\,\Theta_{\left(ab\right)} &=& F_T\left(T\right) \overset{\ \circ}{G}_{ab}+F_{TT}\left(T\right)\,S_{\left(ab\right)}^{\;\;\;\mu}\,\partial_{\mu} T+\frac{g_{ab}}{2}\,\left[F\left(T\right)-T\,F_T\left(T\right)\right],  \label{1001a}
		\\
		0 &=& F_{TT}\left(T\right)\,S_{\left[ab\right]}^{\;\;\;\mu}\,\partial_{\mu} T, \label{1001b}
	\end{eqnarray}
\end{subequations}
with $\overset{\ \circ}{G}_{ab}$ the Einstein tensor, $\Theta_{\left(ab\right)}$ the EM, $T$ the torsion scalar, $g_{ab}$ the gauge metric, $S_{ab}^{\;\;\;\mu}$ the superpotential (torsion dependent) and $\kappa$ the coupling constant. The TEGR symmetric parts of FEs ($F(T)=T$) will be exactly eqn \eqref{einsteineqns} and the antisymmetric part is trivally satisfied (i.e. $0=0$).

The canonical EM is obtained from $\mathcal{L}_{Source}$ term of eqn \eqref{1000} as:
\begin{align}\label{1001ca}
	\Theta_a^{\;\;\mu}=\frac{1}{h} \frac{\delta \mathcal{L}_{Source}}{\delta h^a_{\;\;\mu}}.
\end{align}
The eqn \eqref{1001ca} symmetric parts are \cite{Coley:2024tqe,Landry:2024dzq,Landry:2024pxm}:
\begin{equation}\label{1001c}
	\Theta_{(ab)}= T_{ab},
\end{equation}
where $T_{ab}$ is the canonical EM tensor and the antisymmetric parts are $\Theta_{[ab]}=0$. The eqn \eqref{1001c} is valid only when the matter field interacts with the metric $g_{\mu\nu}$ defined from the coframe $h^a_{\;\;\mu}$ and the gauge $g_{ab}$, and is not directly coupled to the $F(T)$ TPG. This consideration is possible only for the zero hypermomentum case (i.e. $\mathfrak{T}^{\mu\nu}=0$) as discussed in refs. \cite{Golovnev:2020las,Landry:2024dzq,Landry:2024pxm}. This quantity is defined from eqns \eqref{1001a} and \eqref{1001b} components as \cite{Golovnev:2020las}:
\begin{align}\label{1001h}
	\mathfrak{T}_{ab}=\kappa\Theta_{ab}-F_T\left(T\right) \overset{\ \circ}{G}_{ab}-F_{TT}\left(T\right)\,S_{ab}^{\;\;\;\mu}\,\partial_{\mu} T-\frac{g_{ab}}{2}\,\left[F\left(T\right)-T\,F_T\left(T\right)\right].
\end{align}
The EM conservation law in TPG for $\mathfrak{T}^{\mu\nu}=0$ case states that $\Theta_a^{\;\;\mu}$ satisfies the relation \cite{Aldrovandi_Pereira2013,Bahamonde:2021gfp}:
\begin{align}\label{1001e}
	\overset{\ \circ}{\nabla}_{\nu}\left(\Theta^{\mu\nu}\right)=0,
\end{align}
with $\overset{\ \circ}{\nabla}_{\nu}$ the covariant derivative in GR and $\Theta^{\mu\nu}$ the conserved EM tensor. This eqn \eqref{1001e} is also the GR EM conservation laws. Eqn \eqref{1001e} follows from the $\mathfrak{T}^{\mu\nu}=0$ condition (null hypermomentum). For non-zero hypermomentum situations (i.e. $\mathfrak{T}^{\mu\nu}\neq 0$), we need to satisfy more complex relations than eqn \eqref{1001e} \cite{Golovnev:2020las}.

\subsection{Null and constant torsion scalar spacetimes} \label{sect222}

To solve the TppW set of spacetimes, we must discuss the effect of a constant torsion scalar $T=T_0$. Spacetimes which admit this condition include the VSI and CSI spacetimes. For such spacetimes, eqns \eqref{1001a} and \eqref{1001b} will simplify by setting $\partial_{\mu} T=0$ and leaving only the FEs symmetric parts (eqn \eqref{1001a}) \cite{Landry:2023egz}:
\begin{eqnarray}\label{231}
	\kappa \Theta_{\left(ab\right)}\,&=& F_{T}\left(T_0\right)\,\overset{\ \circ}{G}_{ab}+\frac{g_{ab}}{2}\,\left[F\left(T_0\right)-T_0\,F_{T}\left(T_0\right)\right],
\end{eqnarray}
where $\kappa_{eff}=\frac{\kappa}{F_{T}\left(T_0\right)}$. The FEs antisymmetric parts as defined by eqn \eqref{1001b} becomes identically satisfied. We divide eqn \eqref{231} by $F_{T}\left(T_0\right)$ to obtain:
\begin{eqnarray}\label{231a}
	\kappa_{eff} \Theta_{\left(ab\right)}\,&=& \overset{\ \circ}{G}_{ab}+g_{ab}\,\left[\frac{F\left(T_0\right)}{2\,F_{T}\left(T_0\right)}-\frac{T_0}{2}\right].
\end{eqnarray}
We observe that if $T=T_0=\text{Const}$, then the $F(T)$ teleparallel FEs reduce to those of GR as required for TppW spacetimes.

However, we need to consider the $T_0=0$ case for some TppW class of solutions. The FEs reduce to \cite{Landry:2023egz}:
\begin{eqnarray}\label{232}
	\kappa_{eff} \Theta_{\left(ab\right)}\,&=& \overset{\ \circ}{G}_{ab}+g_{ab}\,\left[\frac{F\left(0\right)}{2\,F_{T}\left(0\right)}\right] ,
\end{eqnarray}
where $\kappa_{eff}=\frac{\kappa}{F_{T}\left(0\right)}$. Here $F(0) \neq 0$ for non-vacuum solutions, especially for TppW solution classes. By setting $F(0) = 0$, we exactly recover eqns \eqref{einsteineqns} in GR and hence the TEGR FEs.

\subsection{Symmetry methods} \label{sec:Sym}

As the ppW spacetimes admit a covariant constant null direction, $\bell = \partial_v$, this is a KVF and we can restrict the class of all VSI$_{T}$ geometries by imposing that these geometries admit an affine frame symmetry generator, ${\bf X} = \bell = \partial_v$, satisfying: 

\beq \mathcal{L}_{\bf X} \bh^a = \lambda^a_{~b} \bh^b \text{ and } \mathcal{L}_{\bf X} \omega^a_{~bc} = 0 , \eeq

\noindent where $\lambda^a_{~b}$ is a Lie algebra generator of the Lie group $SO(1,3)$. If there are additional symmetries, with a potentially non-trivial isotropy subgroup, the above equation is difficult to use to determine the other symmetries. To do so, we must introduce the concept of a symmetry frame.

\begin{defn} \label{defn:SymFrame}
	Given a Riemann-Cartan geometry, $(\tilde{\bh}^a, \bomega^a_{~b})$ admitting an affine frame symmetry generator, ${\bf X}$. The class of {\bf symmetry frames}, $\bh^a$ satisfy the following conditions:
	
	\beq && \mathcal{L}_{{\bf X}} \bh^a = f_X^{\hat{i}} \lambda^a_{\hat{i}~b} \bh^b, \label{TP:frm:sym} \eeq
	
	\noindent where $\lambda^a_{\hat{i}~b}$ are basis elements of the Lie algebra of the linear isotropy group (so that $\lambda_{\hat{i} ab} =- \lambda_{\hat{i} ba}$) and $\hat{i}$ ranges from 1 to the dimension of the linear isotropy group. The components of the vector $f_X^{\hat{i}}$ are functions dependent on the coordinates affected by the vector-fields that generate the symmetries.

\end{defn}

\noindent further conditions on the frame and (potentially non-zero) spin-connection arise from the following theorem:

\begin{thm}\label{thm12}
	The most general Riemann-Cartan geometry admitting a given group of symmetries, ${\bf X}_I,~ I,J,K \in \{1, \ldots, N\}$ with a non-trivial isotropy group of dimension $n$ can be determined by solving for the unknowns $h^a_{~\mu}$, $f_I^{~\hat{i}}$ (with $\hat{i}, \hat{j}, \hat{k} \in \{1,\ldots, n\}$) and $\omega^a_{~bc}$ from the following equations:
	
	\beq \begin{aligned} 
		& { X}_I^{~\nu} \partial_{\nu} h^a_{~\mu} + \partial_{\mu} { X}_I^{~\nu} h^a_{~\nu} = f_I^{~\hat{i}} \lambda^a_{\hat{i}~b} h^b_{~\mu} , \\
		& 2{\bf X}_{[I} ( f_{J]}^{~\hat{k}}) - f_I^{~\hat{i}} f_J^{~\hat{j}} C^{\hat{k}}_{~\hat{i} \hat{j}} = C^K_{~IJ} f_K^{~\hat{k}} , \\
		& { X}_I^{~d} \bh_d( \omega^a_{~bc}) + \omega^d_{~bc} f_I^{~\hat{i}} \lambda^a_{\hat{i}~d} - \omega^a_{~dc} f_I^{~\hat{i}} \lambda^d_{\hat{i}~b} - \omega^a_{~bd} f_I^{~\hat{i}} \lambda^d_{\hat{i}~c} - \bh_c( f_I^{~\hat{i}}) \lambda^a_{\hat{i}~b} = 0 , \end{aligned} \label{Sym:RC:Prop} 
	\eeq
	
	\noindent where $\{ \lambda^a_{\hat{i}~b}\}_{\hat{i}=1}^n$ is a basis of the Lie algebra of the isotropy group, $[\lambda_{\hat{i}}, \lambda_{\hat{j}}] = C^{\hat{k}}_{~\hat{i}\hat{j}} \lambda_{\hat{k}}$, $[{\bf X}_I, {\bf X}_J] = C^K_{~IJ} {\bf X}_K$. 
	
\end{thm}

To specialize this result to a teleparallel geometry, in addition to the differential equations (DEs) coming from eqn \eqref{Sym:RC:Prop}, the connection must be flat and so the  curvature tensor expressed in terms of the spin connection and its derivatives must vanish.  

\begin{thm} 
	The most general teleparallel geometry which admits a given group of symmetries, ${\bf X}_I,~ I,J,K \in \{1, \ldots N\}$ with a non-trivial isotropy group of dimension $n$ can be determined by solving for the unknowns $h^a_{~\mu}$, $f_I^{~\hat{i}}$ (with $\hat{i}, \hat{j}, \hat{k} \in \{1, \ldots n\}$) and $\omega^a_{~bc}$ from eqns \eqref{Sym:RC:Prop} along with an additional condition:

	\beq \begin{aligned} R^a_{~bcd} & = \bh_c [\omega^a_{~bd}] - \bh_d [ \omega^a_{~bc}]+ \omega^e_{~bd} \omega^a_{~ec} - \omega^e_{~bc}\omega^a_{~ed}- c^e_{~cd} \omega^a_{~be} = 0 , \end{aligned} \label{Sym:Prop} 
	\eeq
	
	\noindent where $[\bh_c, \bh_d] = c^e_{~cd} \bh_e$ and $c^e_{~cd}$ are the coefficients of anholonomy for the frame $\bh_a$. 
	
\end{thm}

\noindent In the most general class of ppWs, the isotropy subgroup is trivial and this implies that we can set $f_X^{\hat{i}} = 0$. However, for TppW solutions with a larger isometry group,  these may potentially admit a non-trivial linear isotropy group and so it is possible that $f_X^{\hat{i}} \neq 0$.

\subsection{The modified Cartan-Karlhede algorithm for torsion} 

A modification of the CK algorithm was introduced for teleparallel geometries in \cite{Coley:2019zld}. In this section we will briefly review the approach and define terminology that will be used in the remainder of the paper. We will denote $\mathcal{T}^q$ as the set of components of the torsion tensor and the covariant derivatives of the torsion tensor up to the $q$-$th$ covariant derivative: 
\beq \mathcal{T}^q = \{ T_{abc}, T_{abc|d_1}, \ldots T_{abc|d_1 \ldots d_q} \}. \eeq   For any teleparallel geometry the CK algorithm may be summarized by the following steps:

\begin{enumerate}
	\item Set the order of differentiation $q$ to 0.
	\item Calculate $\mathcal{T}^q$.
	\item Determine the canonical form of the $q$-th covariant derivative of the torsion tensor.
	\item Fix the frame as much as possible, using this canonical form, and record the remaining frame transformations that preserve this canonical form (the group of allowed frame transformations is the {\it linear isotropy group $H_q$}). The dimension of $H_q$ is the dimension of the remaining {\it vertical} freedom of the frame bundle.
	\item Find the number $t_q$ of independent functions of spacetime position in $\mathcal{T}^q$ in the canonical form. This tells us the remaining {\it horizontal} freedom.
	\item If the dimension of $H_q$ and number of independent functions are the same as in the previous step, let $p+1=q$, and the algorithm terminates; if they differ (or if $q=0$), increase $q$ by 1 and go to step 2. 
\end{enumerate}

\noindent The resulting non-zero components of $\mathcal{T}^p$ constitute the {\it Cartan invariants} and we will denote them as $\mathcal{T} \equiv \mathcal{T}^{p+1}$ so that:
\beq \mathcal{T} = \{ T_{abc}, T_{abc|d_1}, \ldots T_{abc|d_1 \ldots d_{p+1}} \}. \label{CartanSet} \eeq 

This algorithm will pick out a preferred frame for a given teleparallel geometry which is adapted to the form of the torsion tensor and its covariant derivatives. This motivates the following definitions:

\begin{defn}
	Any frame determined in a coordinate-independent manner using the Lorentz frame transformations to fix the torsion tensor and its covariant derivatives (hereafter, we will refer to these tensors as torsion tensors) into canonical forms is called an {\bf invariantly defined Cartan frame up to linear isotropy $H_p$}. 
\end{defn}

In general, there is non-trivial linear isotropy and so the frame may not be completely determined. We note that the existence of non-trivial linear isotropy in the CK algorithm is necessary and sufficient for a non-trivial isotropy subgroup in the group of affine frame symmetries \cite{Coley:2019zld}. If the linear isotropy group is trivial, $H_{p} = \emptyset$ then the Cartan frame is fully determined by the teleparallel geometry. 

\begin{defn} 
	If a Cartan frame can be constructed by fixing all of the parameters of the  Lorentz frame transformations, then such a frame is called an {\bf invariant frame}: 
\end{defn}

\begin{rem}
	We remark that a {\it proper frame} is not necessarily a {\it Cartan frame} as its definition relies on the vanishing of the spin-connection which is not Lorentz covariant. It is possible to choose the proper frame and then completely fix the remaining Lorentz parameters.
\end{rem}

For sufficiently smooth frames and spin-connection, the result of the algorithm is a set of Cartan scalars providing a unique local geometric characterization of the teleparallel geometry. The 4D spacetime is characterized by the canonical form used for the torsion tensors, the two discrete sequences arising from the successive linear isotropy groups and the independent function counts, and the values of the (non-zero) Cartan invariants. As there are $t_p$ essential spacetime coordinates, the remaining $n-t_p$ are ignorable, and so the dimension of the affine frame symmetry isotropy group (hereafter called the isotropy group) of the spacetime will be $s=\dim(H_p)$ and the affine frame symmetry group has dimension: \beq r=s+n-t_p. \label{rnumber} \eeq

\section{Teleparallel pp-Wave spacetime geometry structure} \label{sec:TppWDef}

For TppW solutions, we need to define the VSI$_T$ teleparallel geometries class for the proper coframe \cite{Ehlers:1962zz,sippel1986symmetry,kramer}:

\beq \begin{aligned} \bell & = du  , \\
	\bn &= dv + H(u,v, x^1, x^2)  du + W_i(u,x^1,x^2) dx^i , \\
	\bm &= M_0(u,x^1, x^2) du + M_1(u) (dx^1 + i dx^2) , \end{aligned} \label{VSIFrame} \eeq
\noindent where $M_0$ and $M_i$ are $\mathbb{C}$-valued functions, $H$ and $W_i$ are $\mathbb{R}$-valued functions and  $H$ is linear in the $v$-coordinate:
\beq H = H^{(1)}(u,x^1,x^2) v + H^{(0)}(u,x^1,x^2). \label{VSIFns} \eeq

\noindent As coordinates we will use $x^1 = x$ and $x^2 = y$. At times it will be useful to use polar coordinates where 
\beq x^1 = x = r \cos (\theta), \quad x^2 = y = r \sin(\theta). \label{PolarCoords} \eeq

Here, the spin-connection components are all $\omega^a_{~b\mu}=0$ in the proper frame. As it is required that $\nabla \bell =0 $ with respect to the LC connection, $\bell$ is a KVF and we may use theorem \ref{thm12} with $f^{\hat{i}}_{\bell} = 0$ to determine solutions that admit $\bell$ as an affine frame symmetry generator. Doing so, we find that the frame functions in eqn \eqref{VSIFrame} are now independent of $v$.  To satisfy eqns \eqref{ppwavedefinition}, it follows that $M_1= P(u) e^{i Q(u)} = 1$. The frame functions in \eqref{VSIFrame} are defined as $H(u,x, y)$, $W_i(u,x,y )$, $M_1 = 1$ and $M_0(u, x,y )$. We have used the fact that a TppW spacetime must admit a null KVF $\bell = \partial_v$ and that it is VSI with regards to torsion. 

To continue we will impose the TEGR FEs:
\beq D_c (F_{T} S_a^{~bc}) + F_{T}[S_a^{~bc} T^d_{~cd} + \frac12 S_a^{~cd} T^b_{~cd} + S_d^{~bc} T^d_{~ca}] + \frac12 \delta_a^{~b} F(T) =  \kappa  \Theta_a^{~b} , \eeq

\noindent where the superpotential is given as: 
\begin{equation}
	S_a^{\phantom{a}b c}=\frac{1}{2}\left(T_a^{\phantom{a}b c}+T^{cb }_{\phantom{cb }a}-T^{b c}_{\phantom{b c}a}\right)-h_a^{\phantom{a}c}T^{d b }_{~~d} + h_a^{\phantom{a}b }T^{d c}_{~~d}.
\end{equation} 

By setting $F(T) = T$ for TEGR, the resulting FEs can be written in a concise form by working with complex spatial coordinates $\zeta = x+iy$ and using the expressions $W = (W_1 -\sqrt{2} Re(M_0)) + i (W_2 -\sqrt{2} Im(M_0))$. The non-trivial TEGR FEs components are:
\begin{subequations}
\begin{align} 
\kappa \Theta_{23} = & \bar{W}_{,\zeta , \overline{\zeta}} - W_{, \overline{\zeta} , \overline{\zeta}}, \label{TEGRFEsa}\\
	\kappa \Theta_{24} =& {W}_{,\overline{\zeta} , {\zeta}} - \bar{W}_{, {\zeta} , {\zeta}}, \label{TEGRFEsb}\\
	\kappa \Theta_{22} =& \Delta' H + 2\bar{M}_0 ( \bar{W}_{,\zeta , \overline{\zeta}} - W_{, \overline{\zeta} , \overline{\zeta}})  + 2 M_0 ( {W}_{,\overline{\zeta} , {\zeta}} - \bar{W}_{, {\zeta} , {\zeta}}) - \frac12 W_{,\zeta} \bar{W}_{\overline{\zeta}} , \label{TEGRFEsc}
\end{align} 
\end{subequations}
\noindent where $\Delta' = \partial_{\zeta} \partial_{\overline{\zeta}}$ is the 2D Laplacian for $x$ and $y$. From eqns \eqref{TEGRFEsa}--\eqref{TEGRFEsc}, we exploit the fact that the existence of a covariant constant null vector, with respect to the LC connection, implies that electromagnetic non-null fields, perfect fluids and $\Lambda$-term solutions cannot occur as sources to the EFEs \cite{kramer} and require this to the teleparallel analogues as well. The only allowed solutions for eqns \eqref{TEGRFEsa}--\eqref{TEGRFEsc} are vacuum, Einstein-Maxwell null fields and pure radiation fields.

This implies that $\Theta_{23} = \Theta_{24} = 0$ and eqns \eqref{TEGRFEsa} and \eqref{TEGRFEsb} are vanishing and then $W = \partial_{\zeta} w$ for a scalar function $w(u,x,y)$. Without loss of generality, we can apply a coordinate transformation $v' = v+w$ so that $W = 0$ in the new coordinate system \cite{pravda2002all}. From this last consideration, we have the following identities:
\beq W_1 = \sqrt{2}\, Re(M_0) \text{ and } W_2 = \sqrt{2}\, Im(M_0), \label{eqn:ppwave2}\eeq

\noindent and then eqn \eqref{TEGRFEsc} will become for $H$:
\beq  \kappa \Theta_{22} = \Delta' H=\left(\partial_x^2+\partial_y^2\right)\,H. \eeq
\noindent With the conditions that $H=H(u,x,y), M_1 =1$ and $W_1 = W_2 = 0$ we have recovered the proper frame for the TppW spacetimes. The non-null torsion tensor components are given in appendix \ref{sec:appendixA}.


\section{The teleparallel analogue of pp-wave solutions with small symmetry groups} \label{sec:TppWsmall}

Noting that the affine frame symmetries are necessarily KVFs but not all KVFs are affine frame symmetries. We will initially use ref. \cite{sippel1986symmetry} to establish the existence of lower symmetry TppW geometries.

In this case, instead of determining an affine symmetry frame, we will work with the following frame independent quantities defined in eqn \eqref{Intro3}, that is the Lie derivatives of the torsion tensor and its covariant derivatives. We will, however, work with the coframe \eqref{VSIFrame} and impose the TppW conditions in eqns \eqref{Intro3}. As the proper frame for the TppW is completely determined, we may use the Lie derivatives with respect to the KVFs to determine constraints on the functions $H$ and $M_0$. Since all affine frame vector fields are KVFs but not all KVFs will be affine frame vector fields, for small symmetry groups we may check the KVFs in \cite{sippel1986symmetry} to determine which are affine frame symmetries. As a final note, for the sake of brevity, we will only display the final result for each symmetry group as the resulting DE coming from \eqref{Intro3} are straightforward to solve.  
 
\subsection{2-dimensional symmetry groups}

In the coordinate system $\{ u,v, r, \theta\}$, in addition to $\bx_1 = \partial_v$ there are three possible $\bx_2$ cases:

\beq \begin{aligned}
\text{Case 2:~ } & \partial_{\theta}, \\
\text{Case 3:~ } & \epsilon \partial_{\theta}+ u \partial_u - v \partial_v, \\
\text{Case 4:~ } & \epsilon \partial_{\theta} + \partial_u.   
\end{aligned} \eeq

\begin{itemize}

\item {\bf $G_2$ Case 2}

The solution to the eqns \eqref{Intro3} is:

\beq
\begin{aligned}
	H &= H(u,r),\\
	M_0 &= m_0(u,r)e^{i\theta} - i \frac{m_1(u)}{r}e^{i\theta}+ C_1, \quad C_1 \in \mathbb{C} ,\\
\end{aligned} \label{eqn:Case2}
\eeq
\noindent where $m_0$ and $m_1$ are $\mathbb{C}$-valued functions.

\item {\bf $G_2$ Case 3}

The solution to the eqns \eqref{Intro3} is:

\beq
\begin{aligned}
	H & = u^{-2} Re(H'(r, u^{- i \epsilon}e^{i \theta})),\\
	M_0 & = u^{i\epsilon-1} m_0(r, u^{- i \epsilon}e^{i \theta}) + C_1 u^{i\epsilon-1} + C_2, \quad C_1, C_2 \in \mathbb{C}, \\
\end{aligned} \label{eqn:Case3}
\eeq

\noindent where $H'$ and $m_0$ are $\mathbb{C}$-valued functions.

\item {\bf $G_2$ Case 4}

The solution to the eqns \eqref{Intro3} is:

\beq
\begin{aligned}
	H &= Re(H'(r, e^{i (\theta-\epsilon u)})),\\
	M_0 &= e^{i\epsilon u} m_0(r, e^{i (\theta-\epsilon u)}) + C_1 e^{i \epsilon u} + C_2, \quad C_1, C_2 \in \mathbb{C} , \\
\end{aligned} \label{eqn:Case4}
\eeq

\noindent where $H'$ and $m_0$ are $\mathbb{C}$-valued functions.

\end{itemize}

\subsection{3-dimensional symmetry groups}

In the coordinate system $\{ u,v, r, \theta\}$, in addition to $\bx_1 = \partial_v$ there are four possible pairs of KVFs $\left(\bx_2,\,\bx_3\right)$:
\beq \begin{aligned}
    \text{Case 5:~} & u \partial_u-v \partial_v, \text{ and } \partial_{\theta},  \\
    \text{Case 6:~} & \partial_u \text{ and } \partial_{\theta}, \\
    \text{Case 7:~} & \partial_u \text{ and } \epsilon \left( u \partial_u - v \partial_v\right)-\partial_{\theta}, \\ 
    \text{Case 8:~}  & \partial_u \text{ and } \epsilon \left( u \partial_u - v \partial_v\right)+\rho\,\partial_x+\sigma\,\partial_y, \\
\end{aligned} \eeq, 

\noindent where $\epsilon, \rho$ and $\sigma$ are $\mathbb{R}$-valued constants.

\begin{itemize}

\item {\bf $G_3$ Case 5}

The solutions to the eqns \eqref{Intro3} is:
\beq
\begin{aligned}
    H & = \frac{H_0(r)}{u^2},\\
    M_0 &= \frac{m_0(r)}{u}\,e^{i\theta} + C_1, \quad C_1 \in \mathbb{C} ,
\end{aligned} \label{eqn:Case5}
\eeq

\noindent where $m_0$ is a $\mathbb{C}$-valued function.

\item {\bf $G_3$ Case 6}

The solutions to the eqns \eqref{Intro3} is:
\beq
\begin{aligned}
    H & = H_0(r),\\
    M_0 &= m_0(r)\,e^{i\theta}+ C_1, \quad C_1 \in \mathbb{C} ,
\end{aligned} \label{eqn:Case6}
\eeq

\noindent where $m_0$ is a $\mathbb{C}$-valued function.

\item {\bf $G_3$ Case 7}

The solutions to the eqns \eqref{Intro3} is:
\beq
\begin{aligned}
    H & = H_0(r)\,e^{2\epsilon\theta},\\
    M_0 &= m_0(r)\,e^{(i+\epsilon)\theta}+ C_1, \quad C_1 \in \mathbb{C},
\end{aligned} \label{eqn:Case7}
\eeq

\noindent where $m_0$ is a $\mathbb{C}$-valued function.

\item {\bf $G_3$ Case 8}

There are two solution to the eqns \eqref{Intro3}.
If $\epsilon = 0$ then the functions are: 
\beq
\begin{aligned}
    H & = H_0\left(\rho\,y-\sigma x\right),\\
    M_0 &= C_1 u + C_2, \quad C_1, C_2 \in \mathbb{C}.
\end{aligned} \label{eqn:Case8a}
\eeq
 Otherwise, with $\epsilon \neq 0$ the functions take the form:
\beq
\begin{aligned}
    H & = H_0\left(\rho\,y-\sigma x\right)\,e^{-\frac{2\epsilon}{\rho}\,x},\\
    M_0 &= m_0\left(\rho\,y-\sigma x \right)\,e^{-\frac{\epsilon}{\rho}\,x}+ C_1, \quad C_1 \in \mathbb{C},
\end{aligned} \label{eqn:Case8b}
\eeq
where $\rho\neq 0$ and $\sigma \neq 0$ and $m_0$ is a $\mathbb{C}$-valued function.
\end{itemize}

\section{Teleparallel pp-Wave solutions admitting linear isotropy}\label{sec:isotropy}

For solutions with larger symmetry groups in GR, it is difficult to determine the potential subgroup of symmetries that hold for its teleparallel analogue. In particular, we anticipate the 5, 6 and 7-dimensional (5D, 6D and 7D) symmetry groups determined in \cite{sippel1986symmetry}  have significantly smaller symmetry groups for their teleparallel analogue. 

However, the analysis carried out so far is difficult to impliment due to these cases potentially requiring the solution of a system of second order DEs for two unknowns which depend on 2 arbitrary functions. In order to study these solutions we will have to use a different set of tools coming from the CK algorithm.  Since we have an explicit form for the TppW solutions, we can determine explicitly solutions that admit linear isotropy and hence determine all TppW solutions that admit a $G_5$ or higher.

There are three irreducible parts of the torsion tensor that are relevant for this discussion, the vector-part, $V_a = T^b_{~ba}$, the axial-part, $A_a = 
\epsilon_a^{~bcd} T_{bcd}$ and the tensor part, $
\hat{T}^a_{~bc}$ such that $
\hat{T}^a_{~ab} = 
\hat{T}_{[abcd]} = 0$. The linear isotropy group for the first two parts will be a subgroup of the Lorentz group that leaves the vector (or one-form) invariant. Since $V_a$ and $A_a$ are proportional to $
\bell$, the linear isotropy group for both parts is $E(2)$. Relative to the frame this subgroup takes the form:
\beq \begin{aligned}
    & \text{Boosts: } \bell' = D^2 \bell,~~{\bf n'} = D^{-2} {\bf n}, {\bf m}'=  {\bf m},~~ D, \in \mathbb{R},\\
    &\text{ Spins: } \bell' = \bell,~~{\bf n'} =   {\bf n}, {\bf m}'= e^{i \theta} {\bf m},~~ \theta \in \mathbb{R}, \\
    &\text{Null Rotations: } \bell' = \bell,~~{\bf n}' = {\bf n} + B {\bf m} + \bar{B} \bar{{\bf m}} + B \bar{B} \bell, {\bf m}' = {\bf m} + B \bell,~~ B \in \mathbb{C}. 
\end{aligned} \label{eqn:SE2} \eeq
\noindent The permitted linear isotropy groups of the tensor part has been studied in \cite{Coley:2019zld}. In what follows we will use these Lorentz transformations to determine the permitted linear isotropy groups for $\nabla^p T^a_{~bc}$ for $p>0$.

To begin we will consider the possibility of the entirety of $SE(2)$ acting as the linear isotropy group. This occurs when $\hat{T}^a_{~bc}$ vanishes. If $\hat{T}^a_{~bc}$ is non-zero, the potential linear isotropy groups are \cite{Coley:2019zld}:
\begin{itemize}
    \item 1-dimensional (1D) boost isotropy,
    \item 1D spin isotropy,
    \item 1D null rotation isotropy,
    \item 2D null rotations isotropy.
\end{itemize}
\noindent The group $SE(2)$ is excluded from the TppW spacetimes due to the following proposition.
\begin{prop}
    If $\hat{T}^a_{~bc} = 0$ then the resulting teleparallel geometry is Minkowski space.
\end{prop}

\begin{proof}
    This follows by setting the components of $\hat{T}^a_{~bc}$ to zero and solving for $M_0$ and $H$, yielding the following solution:
    \beq M = M_0(u),\quad H = Re((x+iy)M_{0,u}) + h(u). \eeq
    However, putting this solution into $A_a$ and $V_a$, it follows that both of these vanish, and hence the torsion tensor vanishes. This can only occur for Minkowski space. 
\end{proof}

Similarly, boost isotropy is not permitted for the TppW spacetimes. 
\begin{prop}
    If a TppW admits boost isotropy then it is Minkowski space.
\end{prop}

\begin{proof}
    A tensor admits boost isotropy if the only non-zero components of the tensor always contain a $1$ paired with a $2$ of the same type. Imposing this condition for irreducible parts of the torsion tensor, it follows that $\hat{T}$ must vanish and hence is necessarily Minkowski space. 
\end{proof}

Similarly, the 2D linear isotropy group is not permitted for TppW spacetimes.

\begin{prop}
    If a TppW admits a 2D group of null rotations with complex parameter $B$ then it is Minkwoski space. 
\end{prop}

\begin{proof}
    At zeroth order, in order to permit such a linear isotropy group the following components must vanish:
    \beq \hat{T}_{323} = \hat{T}_{212} = \hat{T}_{324} = 0. \eeq 
    From Appendix A, this implies that:
    \beq M_0 = m_0(u). \eeq
    At first order, the linear isotropy condition implies that:
    \beq \begin{aligned}
        H_{,x,y}=H_{,x,x}=H_{,y,y} = 0  .
    \end{aligned} \eeq
    This condition forces the torsion tensor to vanish identically, implying that the spacetime is in fact Minkowski space.
\end{proof}

The remaining 1D linear isotropy groups do yield non-trivial TppW spacetimes. 

\begin{prop}
    If a TppW admits spin isotropy then the frame functions are of the form:
    \beq 
        M_0 = (C_0(u)+iC_1(u))(x+iy), H=C_2(u)(x^2+y^2) ,
    \eeq
    where $C_0, C_1$ and $C_2$ satisfy the following DEs:
    \beq
        \begin{aligned}
            & C_{2,u} - \left(\frac{3 C_{1,u}}{C_1} + 2\sqrt{2} C_0 \right) C_2 - \frac{C_0 C_{1,u,u}}{\sqrt{2} C_1} - \frac{C_{0,u,u}}{\sqrt{2} } = 0 ,\\ 
            & C_{2,u} - \frac{(3 C_{0,u} -\sqrt{2}C_2 -\sqrt{2}C_0^2  + \sqrt{2} C_1^2)}{C_0} C_2  - \frac{\sqrt{2}C_{1,u}^2}{C_0}+ \frac{\sqrt{2} C_{0,u}^2}{C_0} + \frac{ C_1C_{1,u,u}}{\sqrt{2}C_0}- \frac{C_{0,u,u}}{\sqrt{2}}= 0.
        \end{aligned}
    \eeq
    Furthermore, the symmetry group for such a TppW spacetime admits a 4D symmetry group.
    
\end{prop}

\begin{proof}
    At zeroth order, the torsion tensor admits spin isotropy only if :
    \beq 
        \hat{T}_{223}= \hat{T}_{323} =0. 
    \eeq
    However, the first condition can be achieved by a null rotation which fully fixes $B \in \mathcal{C}$. Similarly, applying a boost with parameter $D$ to set $\hat{T}_{212} =1$ may always be chosen. The resulting frame is invariantly defined up to spin linear isotropy. The vanishing of $\hat{T}_{323}$ in Appendix A implies that: 
    \beq 
        M_0 = m_0(u, x+iy) ,
    \eeq
    for some $\mathbb{C}$-valued function $m_0$. At first order, all components without 3 and 4 paired must vanish. In particular, $\hat{T}_{333;2}= 0 $ and $\hat{T}_{322;3} = 0$ imply that 
    \beq M_0 = m_0(u)(x+iy), H = C_2(u) (x^2+y^2). \eeq
    \noindent There are two further conditions coming from the vanishing of $\hat{T}_{222;3}$. The real and imaginary parts can be combined to yield the required DEs in terms of $C_0, C_1$ and $C_2$. 

    In the chosen invariantly defined frame, the non-zero components of the torsion tensor and its covariant derivative are functions of $u$. Imposing the above DEs implies that the second covariant derivative of the torsion tensor also admits spin linear isotropy and its components are also functions of $u$ alone. Thus the number of functionally independent components and the dimension of the isotropy group has stabilized at second order and the CK algorithm terminates. The symmetry group is then according to eqn \eqref{rnumber}: $r=1+4-1 = 4$. 
\end{proof}

\noindent For the 1D group of null rotations, it is always possible to choose coordinates and the frame so that the null rotation parameter is $\mathbb{R}$-valued or purely imaginary. In what follows we will choose the null rotation parameter to be $\mathbb{R}$-valued. 

\begin{prop}
    If a TppW admits a 1D group of null rotations with parameter $\bar{B} =B$ then the frame functions are of the form:
    \beq 
        M_0 = M_0(u, y), \quad\quad\quad H = h_1(u) x + h_2(u,y). \label{eqn:CasemissingG3}
    \eeq
\end{prop}

\begin{proof}
    At zeroth order, the linear isotropy condition implies that:
    \beq 
        \hat{T}_{324}+\hat{T}_{424} = M_{0,x} = 0 .
    \eeq
    \noindent To fix the frame we may choose $D$ and $\theta$ such that $\hat{T}_{323}=1$ and fix $\hat{B}'=- B'$ so that $\hat{T}_{223} =\hat{T}_{224}$. The resulting frame is an invariantly defined frame up to the 1D linear isotropy group. At first order, applying a null rotation with parameter $\bar{B} = B$ we find that: \beq 
        \hat{T}_{222;3} = H_{,x,y} + H_{,x,x} = 0 .
    \eeq
    \noindent The solution to this DE is $H = h_1(u)x + h_2(u,y)$. 

    Imposing these conditions it is straightforward to show that the second covariant derivative of the torsion tensor admits the same linear isotropy group. Relative to the invariantly defined frame, there are at most two functionally independent invariants with one appearing at zeroth order and the second at first order. These observations imply that the CK algorithm stops at second order. Using eqn \eqref{rnumber} it follows that the dimension of the symmetry group is $r= 1+4-2 = 3$.
\end{proof}

\begin{rem}
    
As a solution in GR, this is actually a $G_3$ that was missed in ref.~\cite{sippel1986symmetry}, as it admits the KVFs:
\beq 
    \partial_v, - \int h_1 du \partial_v+\partial_x, \left(x-\int u h_1 du\right) \partial_v+u\partial_x.
\eeq

\end{rem}

\subsection{Determining larger symmetry groups with the Cartan-Karlhede algorithm}

It is clear that the dimension of the symmetry group is potentially reduced when one moves from KVF fields in GR to affine frame symmetries in TPG theories due to the constraint in eqn \eqref{Intro3}. A natural question to ask is what is the largest group of affine frame symmetries for the class of teleparallel plane waves. To determine this we will use the CK algorithm to determine the ppW solutions that admit linear isotropy and are locally homogeneous. This condition imposes that relative to an invariantly defined frame, the components of the torsion tensor and its covariant derivatives are constant. 

Using the CK algorithm we find the following results for the TppWs which admit a non-trivial linear isotropy group: 

\begin{prop} \label{prop:1dNullRot}
    The class of locally homogeneous TppWs admitting a
    1D linear isotropy group of null rotations admits a $G_5$ of affine frame transformations and has the following forms:
    \beq
    \begin{aligned}
        & H = C_3 e^{C_2 y}, \quad\quad\quad M_0 = e^{C_2 y}, \quad \\
        & H = -\frac{C_3 C_2 (\sin(C_1/2) + \cos(C_1/2) )}{\cos(C_1)}  e^{C_2 y} , M_0 = e^{C_2 y+i C_1}.
    \end{aligned}        
    \eeq
\end{prop}

\begin{proof}
    Using the conditions in Proposition \ref{prop:1dNullRot} we will choose the parameters $D, \theta$ and $B'$ ($\hat{B}'=- B'$) so that $\hat{T}_{323}=1$ and $\hat{T}_{223} =\hat{T}_{224}$ to produce an invariantly defined frame up to linear isotropy. The geometry is locally homogeneous if and only if the resulting Cartan invariants are constant for each iteration of the CK algorithm. 

    It is straightforward to solve the resulting DEs at zeroth and first order that arise by requiring all Cartan invariants to be constant. Doing so, we recover the required form of $H$ and $M_0$. As there are no functionally independent invariants, according to eqn \eqref{rnumber}, the dimension of the symmetry group is $r=1+4-0 = 5$. 
\end{proof}

\begin{rem}
    The solutions for $H$ are related to Case 9 in ref.~\cite{sippel1986symmetry} through a coordinate change, and this is discussed in the next section.
\end{rem}

\begin{prop} \label{prop:1DSpin}
    The class of locally homogeneous ppWs admitting a
    1D linear isotropy group of spatial rotations admits a $G_5$ of affine frame transformations and has the following two forms:
    \beq
        H = C_1(x^2 +y^2),\quad\quad\quad M_0 = C_1 (x+iy),
    \eeq
    \noindent or 
    \beq
        H = \frac{h}{u^2}(x^2 +y^2),\quad\quad\quad M_0 = \frac{(c_0+ic_1)}{u} (x+iy), 
    \eeq
    \noindent where $h = -c_0^2-\frac{c_0}{\sqrt{2}}$ with $c_1 =0$ when $h \leq \frac{1}{8}$ and $h= \frac{c_1^2+1}{8}$ with $c_0=-\frac{\sqrt{2}}{4u}$ when $h > \frac{1}{8}$ .
\end{prop}

\begin{proof}
    Using the conditions in Proposition \ref{prop:1dNullRot} we will choose the parameters $D, \theta$ and $B$ so that $\hat{T}_{323}=1$ and $\hat{T}_{223} =0$ to produce an invariantly defined frame up to linear isotropy. The geometry is locally homogeneous if and only if the resulting Cartan invariants are constant for each iteration of the CK algorithm. 

    It is straightforward to solve the resulting DEs at zeroth and first order that arise by requiring all Cartan invariants to be constant. Doing so, we recover the required form of $H$ and $M_0$. As there are no functionally independent invariants, according to eqn \eqref{rnumber}, the dimension of the symmetry group is $r=1+4-0 = 5$. 
\end{proof}

\begin{rem} The CK algorithm is able to determine another case that was missed in ref.~\cite{sippel1986symmetry}, namely a $G_4$ solution that is a special case of the $G_3$ solution admitting a 1D group of null rotations.
\end{rem}

\begin{prop}
    The class of ppWs admitting a $G_4$ of affine frame transformations with a 1D linear isotropy subgroup of null rotations has the following forms:
    \beq
        H = h(u + C_0 y),\quad\quad\quad M_0 = m_0 (u+ C_0 y) , \label{eqn:CaseMissingG4}
    \eeq
    where $m_0$ is a $\mathbb{R}$-valued function. The new affine frame symmetry generator is then
    \beq 
        C_0 \partial_u + \partial_y. 
    \eeq
    
\end{prop}

\begin{proof}
    Using the conditions in Proposition \ref{prop:1dNullRot} we will choose the parameters $D, \theta$ and $B'$ ($\hat{B}'=- B'$) so that $\hat{T}_{323}=1$ and $\hat{T}_{223} =\hat{T}_{224}$ to produce an invariantly defined frame up to linear isotropy. If there is only one functionally independent invariant at all iterations of the CK algorithm, then this implies that for any two Cartan invariants, $I$ and $J$, $dI \wedge d J = 0$. Applying this result for all pairs of Cartan invariants at zeroth and first order yields a finite list of DEs.
    
    A tedious exercise in solving these DEs yields the required form of $H$ and $M_0$. As there is only one functionally independent invariant, according to eqn \eqref{rnumber}, the dimension of the symmetry group is $r=1+4-1 = 4$. We note that using the anzatz for $H$ in Proposition \eqref{prop:1dNullRot}, and the DEs in \cite{sippel1986symmetry} it is possible to derive the restricted form of $H$ and the KVF as well.
\end{proof}

\section{The teleparallel analogue of pp-Wave solutions with larger symmetry groups} \label{sec:TppWLarge}

\subsection{5-dimensional symmetry groups}

In the coordinate system $\{ u,v, x, y\}$, in addition to $\bx_1 = \partial_v$ there are two possible sets, $\left( \bx_2, \bx_3,\,\bx_4,\,\bx_5\right)$, of KVFs \footnote{i.e.: There is a sign mistake in the expression for the KVF ${\bx}_4$ for Case 9 in ref.~\cite{sippel1986symmetry}. }:
\beq \begin{aligned}
    \text{Case 9: } & \partial_u, \partial_y+\sigma (u\partial_u-v\partial_v), \partial_x{\bf -}\rho (u\partial_u-v\partial_v) \text{ and }  u(\sigma\partial_x+\rho\partial_y)+(\sigma\,x+\rho\,y)\partial_v ,
    \\
    \,
    \\
    \text{Case 10: } & f_1(u)\,\partial_x+f_2(u)\,\partial_y+(f_1'(u)\,x+f_2'(u)\,y)\,\partial_u, 
\end{aligned} \eeq

\noindent where $\rho, \sigma \in \mathbb{R}$ and the $\mathbb{R}$-valued functions $f_1$ and $f_2$ satisfy a system of two second order DEs leading to four possible solutions pairs. 

As in the previous section, we will only state the results and omit the DEs that arise from eqn \eqref{Intro3}. The only exception will be case 10 where the DEs are necessary to prove that the group of affine frame symmetries is smaller than the the symmetry group of the ppW solution. 

\begin{itemize}

\item {\bf $G_5$ Case 9}

The solutions to the eqns \eqref{Intro3} is:
\beq
\begin{aligned}
    H & = H_0\,e^{2\left(\rho\,x-\sigma\,y\right)},\\
    M_0 &= m_0\,e^{\rho\,x-\sigma\,y}+ C_1, \quad C_1 \in \mathbb{C},
\end{aligned} \label{eqn:case9}
\eeq
where $\rho^2+\sigma^2 \neq 0$.

Using the transformation,
\beq x' = \frac{\rho x + \sigma y }{\sqrt{(\rho^2 + \sigma^2)}}, \quad y' = \frac{\rho y - \sigma y }{\sqrt{(\rho^2 + \sigma^2)}}, \eeq

we recover the locally homogeneous TppW geometry with a 1D null rotation isotropy, which is described in Proposition \ref{prop:1dNullRot}. For this example the teleparallel analogue has the same symmetry group.

\item {\bf $G_5$ Case 10}

To determine the solution, we need to satisfy the following equation system according to ref. \cite{sippel1986symmetry}:
\begin{align}\label{addG5case2system}
f_1(u)\,\partial_{xx}H+f_2(u)\,\partial_{xy} H = -f_1''(u) ,
\nonumber\\
f_1(u)\,\partial_{xy}H+f_2(u)\,\partial_{yy} H = -f_2''(u).
\end{align}
For which the solution for $H(u,x,y)$ is \cite{sippel1986symmetry}: \beq H(u,x,y)=a(u)x^2+ 2c(u)\,xy+b(u)y^2  . \label{G5H0}\eeq
\noindent Substituting this into \eqref{addG5case2system} gives a system of two DEs for $f_1$ and $f_2$ in terms of $A, B$ and $C$:
\begin{align}\label{addG5case2systemalg}
2\left(a(u)\, f_1(u)+c(u)\, f_2(u)\right) = -f_1''(u) ,
\nonumber\\
2\left(c(u)\, f_1(u)+b(u)\, f_2(u)\right) = -f_2''(u).
\end{align}

As $H$ is not of the permitted forms in propositions \ref{prop:1dNullRot} or \ref{prop:1DSpin}, the teleparallel analogue will not admit a non-trivial isotropy subgroup. Thus, it has at most a 4D symmetry group. In order to admit an affine frame symmetry generator of the required form for Case 10, the vanishing of the Lie derivative of the tensor part of the torsion tensor imply that the complex function $M_0$ must be linear in $x$ and $y$:
\beq
\begin{aligned}
       M_0 &= \left(M_1(u)+i\,M_2(u)\right)\,x + \left(M_3(u)+i\,M_4(u)\right)\,y + p_1(u) + i\, p_2(u) . \\
\end{aligned} \label{G5M0}
\eeq
Substituting this into the remaining conditions give the following first order DEs for $f_1$ and $f_2$:
\begin{align}\label{DEsforddee2}
\sqrt{2} \left(M_1(u)\,f_1(u)+M_3(u)\,f_2(u)\right) = -f_1'(u) ,
\nonumber\\
\sqrt{2} \left(M_2(u)\,f_1(u)+M_4(u)\,f_2(u)\right) = -f_2'(u) .
\end{align}

\noindent This can be recast into an inhomogeneous linear second order DE for one function and an algebraic expression for the other. From the uniqueness and existence theorem for second order DEs, this will yield two solutions.

Additional conditions can be imposed on the arbitrary functions of $u$ by differentiating eqns \eqref{DEsforddee2} and substituting them into eqns \eqref{addG5case2systemalg}: 
\beq \begin{aligned}
    & f_2 ( m_4' - \sqrt{2} b -\sqrt{2}m_4^2 -\sqrt{2} m_2 m_3 )  + f_1 (m_2'-\frac{1}{\sqrt{2}}c  -\sqrt{2} (m_4 +m_1) m_2)= 0, \\ 
    &  f_1( m_1' - \sqrt{2} a - \sqrt{2} m_2 m_3  - \sqrt{2} m_1^2 ) + f_2 (m_3' - \frac{1}{\sqrt{2}} c - \sqrt{2} (m_4 + m_1 ) m_3 )=0. \\
\end{aligned} \label{abceqn} \eeq

\noindent These equations must hold for all solutions of $f_1$ and $f_2$ and so the coefficients of $f_1$ and $f_2$ must vanish. In order to determine $c$ algebraically, $m_2$ and $m_3$ must satisfy 
\beq m_3 = m_2 + C_0 e^{\sqrt{2} \int (m_1+m_4) du }.  \label{m2m3eqn} \eeq

Assuming that the functions $a,~b$ and $c$ satisfy the algebraic constraints in \eqref{abceqn}, for arbitrary functions $m_1, m_2$ and $m_4$ then this solution admits at most 2 additional affine frame symmetry generators and hence has a 3-dimensional (3D) symmetry group.

\end{itemize}

\subsection{6-dimensional symmetry groups}

All of the 6D symmetry groups are special cases of the 5D symmetry group in Case 10 of \cite{sippel1986symmetry}. These arise by adding one of the following KVFs:

\beq \begin{aligned}
    \text{Case 11: } & u\,\partial_u-v\,\partial_v ,\\
    \text{Case 12: } & u\,\partial_u-v\,\partial_v+\epsilon (y \partial_x - x \partial_y) ,\\
    \text{Case 13: } & \partial_u,\\
    \text{Case 14: } & \partial_u+\epsilon\,(y \partial_x - x \partial_y),\\
    \text{Case 15: } & y \partial_x - x \partial_y .\\
\end{aligned} \eeq

\begin{itemize}

\item {\bf $G_6$ Case 11}

The solutions to the eqns \eqref{Intro3} are:
\beq
\begin{aligned}
    H & = \frac{a\,x^2+2b x y+ c x^2}{u^2} ,\\
    M_0 &= \frac{ \left(m_1+i\,m_2\right)}{u}\,x + \frac{\left(m_3+i\,m_4\right)}{u}\,y,
\end{aligned} \label{eqn:Case11a}
\eeq
where $a$, $b$, $c$ and $m_i$ ($i=1,2,3,4$) are now constants. Using eqns \eqref{abceqn} and \eqref{m2m3eqn} the following algebraic conditions hold:
\beq \begin{aligned}\label{case11conditions}
    & m_3 = m_2, \\
    & a = -m_1^2 - \frac{m_1}{\sqrt{2}}-m_2^2, \\
    & b = -m_4^2-\frac{m_4}{\sqrt{2}} - m_2^2, \\
    & c = -\sqrt{2} m_2 -2 m_1 m_2 - 2 m_2 m_4. 
\end{aligned} \label{eqn:Case11b} \eeq

\noindent This solution is now locally homogeneous since it admits a 4D symmetry group. 

\item {\bf $G_6$ Case 12}

The solution in this case has the following functions for $H$ and $M_0$ in eqns \eqref{G5H0} and \eqref{G5M0}:

\beq \begin{aligned}
    a & = -\frac{D_0}{u^2}(\sin(2\epsilon \ln|u|)+D_1),\\
    b & = \frac{D_0}{u^2} (\sin(2\epsilon \ln|u|)-D_1),\\
    c & = -\frac{2D_0}{u^2}\cos(2\epsilon \ln|u|),
\end{aligned} \label{eqn:Case12a} \eeq
\beq \begin{aligned}
    m_1 & = \frac{-2(C_1 \sin(2\epsilon \ln|u|) + C_2 \cos(2\epsilon \ln|u|))+ C_3}{2u} ,\\
    m_2 & = \frac{2(C_1 \cos(2\epsilon \ln|u|) + C_2 \sin(2\epsilon \ln|u|))- C_4}{2u} ,\\
    m_3 & = \frac{2(C_1 \cos(2\epsilon \ln|u|) + C_2 \sin(2\epsilon \ln|u|))+ C_4}{2u}  ,\\
    m_4 & = \frac{2(C_1 \sin(2\epsilon \ln|u|) + C_2 \cos(2\epsilon \ln|u|))+ C_3}{2u} ,\\
\end{aligned} \label{eqn:Case12b} \eeq

\noindent where the conditions in eqn \eqref{abceqn} yields two possible branches from:
\beq 
    C_4 (\sqrt{2} C_3 +1) = 0. 
\eeq
\noindent In the first branch, we find the following conditions:
\beq 
    C_3 = - \frac{1}{\sqrt{2}},~C_1 =0,~D_0 = \sqrt{2} \epsilon C_2,~D_1 = -C_2^2+ \frac{C_4^2}{4}+\frac{1}{8}. \label{case12conditions} 
\eeq 
\noindent While in the second branch, we have: 
\beq 
    C_4 = 0, C_1=0, C_2=0, D_0 = 0, D_1 = -\frac{C_3^2}{4} - \frac{\sqrt{2} C_3}{4}.
\eeq

These solutions are now locally homogeneous since they admit a 4D symmetry group. 

\item {\bf $G_6$ Case 13}

The solutions to the eqns \eqref{Intro3} are:
\beq
\begin{aligned}
    H & = a\,x^2+2b x y+ c x^2 ,\\
    M_0 &=  \left(m_1+i\,m_2\right)\,x + \left(m_3+i\,m_4\right)\,y,
\end{aligned} \label{eqn:Case13a}
\eeq
where $a$, $b$, $c$ and $m_i$ ($i=1,2,3,4$) are constants. Using eqns \eqref{abceqn} and \eqref{m2m3eqn} the following algebraic conditions hold:
\beq \begin{aligned}
    & m_3 = m_2, \\
    & a = -m_1^2 -m_2^2, \\
    & b = -m_4^2- m_2^2, \\
    & c = -2 m_1 m_2 - 2 m_2 m_4. 
\end{aligned} \label{eqn:Case13b} \eeq

\noindent This solution is now locally homogeneous since it admits a 4D symmetry group. 

\item {\bf $G_6$ Case 14}

The solution in this case has the following functions for $H$ and $M_0$ in eqns \eqref{G5H0} and \eqref{G5M0}:

\beq \begin{aligned}
    a & = -D_0(\sin(2\epsilon u)+D_1),\\
    b & = D_0 (\sin(2\epsilon u)-D_1),\\
    c & = -2D_0 \cos(2\epsilon u),
\end{aligned} \label{eqn:Case14a} \eeq
\beq \begin{aligned}
    m_1 & = -2(C_1 \sin(2\epsilon u) + C_2 \cos(2\epsilon u))+ C_3,\\
    m_2 & = 2(C_1 \cos(2\epsilon u) + C_2 \sin(2\epsilon u))+ C_4 ,\\
    m_3 & = 2(C_1 \cos(2\epsilon u) - C_2 \sin(2\epsilon u))- C_4 ,\\
    m_4 & = -2(C_1 \sin(2\epsilon u) + C_2 \cos(2\epsilon u))+ C_3.\\
\end{aligned} \label{eqn:Case14b} \eeq

\noindent where the conditions in eqn \eqref{abceqn} yields two possible branches from:
\beq 
    C_3 C_4 = 0 .
\eeq
\noindent The first branch gives:
\beq 
    C_3=0,~ C_1=0,~D_0 = \sqrt{2} \epsilon C_2,~ D_1=- C_2^2+\frac{C_4^2}{4}. 
\eeq 
\noindent The second branch then yields:
\beq 
    &C_4=0, C_1 = \frac{C_2 C_3}{\epsilon \sqrt{2}}, D_0 = \frac{2 \sqrt{2} C_2 C_3+4 \sqrt{2} C_2 \epsilon^2}{4 \epsilon}, D_1 = -\frac{2 C_2^2 C_3^2 + 4C_2^2 \epsilon^2 + C_3^2 \epsilon^2}{4 \epsilon^2}.& \eeq



These solutions are now locally homogeneous since they admits a 4D symmetry group.

\item {\bf $G_6$ Case 15}

In this case, the equations from eqns \eqref{Intro3} yield:

\beq \begin{aligned}
    H &= h(u)(x^2 + y^2), \\
    M_0 &= \left(\frac{m_{,u}}{2 \sqrt{2} m}+i m(u) \right) (x+iy),
\end{aligned} \label{eqn:Case15a} \eeq

\noindent where 
\beq h =  m^2 + \frac{m_{,u,u}}{4 m^2} - \frac{3m_{,u}^2}{8 m^2}. \label{eqn:Case15b} \eeq

This admits a 4D group of affine frame symmetries, where the orbit is 3D with spatial rotations in the $x-y$ plane as an isotropy.

\begin{rem}
    Using the CK algorithm it is possible to show that this simpler solution is equivalent to the more complicated solution in Proposition \eqref{prop:1DSpin}. While an explicit coordinate transformation has not been found by the authors, its existence is guaranteed since the two solutions share the same discrete sequences and form of the Cartan invariants.
\end{rem}
\end{itemize}

\subsection{7-dimensional symmetry groups}

The 7D symmetry groups arise by adding two KVFs to the 5D symmetry group in Case 10. There are two possibilities for this:

\beq \begin{aligned}
    \text{Case 16 : }  & \partial_u \text{ and } y \partial_x - x \partial_y, \\
    \text{Case 17 : }  & u\,\partial_u-v\,\partial_v \text{ and } y \partial_x - x \partial_y . \\
\end{aligned} \eeq

\noindent The teleparallel analogue of these solutions have already been found using the CK algorithm. However, the results in Proposition \ref{prop:1DSpin} can be reproduced using eqn \eqref{Intro3} as well. 

\begin{itemize}

\item {\bf $G_7$ Case 16}

The solution in this case has the following functions for $H$ and $M_0$:
\beq
        H = C_1(x^2 +y^2),~ M_0 = C_1 (x+iy) .
\label{eqn:Case16} \eeq

 The orbits are now 4D and hence is a locally homogeneous solution with a 1D linear isotropy group of spatial rotations. This solutions thus admits a 5D group of symmetries.

\item {\bf $G_7$ Case 17}

The solution in this case has the following functions for $H$ and $M_0$:
    \beq
        H = \frac{h}{u^2}(x^2 +y^2),~ M_0 = \frac{(c_0+ic_1)}{u} (x+iy) . \label{eqn:Case17}
    \eeq
    \noindent In order to cover the range of permitted values of $h$ there are two subcases: 
    \begin{itemize}
        \item $h = -c_0^2-\frac{c_0}{\sqrt{2}}$ with $c_1 =0$ when $h \leq \frac{1}{8}$ ,
        \item $h= \frac{c_1^2+1}{8}$ with $c_0=-\frac{\sqrt{2}}{4u}$ when $h > \frac{1}{8}$ .
    \end{itemize}

 The orbits are now 4D and hence is a locally homogeneous solution with a 1D linear isotropy group of spatial rotations. This solutions thus admits a 5D group of symmetries.

\end{itemize}

\section{Conclusion} \label{sec: Discussion}

Using frame based methods, we have determined all possible symmetry groups for the TppW spacetimes. Most of the 2D and 3D groups found for ppWs in ref.~\cite{sippel1986symmetry} are preserved in the teleparallel analogue (TppW). However, for larger symmetry groups which admit non-trivial isotropy subgroups we have shown that the symmetry group of the teleparallel analogue is significantly smaller. The results of this analysis is summarized in the following table. In the table, the column under $G_d$ denotes the dimension of the symmetry group in GR, while the column under $\tilde{G}_d$, denotes the dimension of the symmetry group in TPG. We note that the orbit of the symmetry group is unchanged when going from GR to TPG.

We have also obtained new solutions that were missed in the original analysis of ppW spacetimes, namely an additional solution admitting a 3D group of symmetries. This solution then leads to a new subcase which admits a 4D group of symmetries. This is novel because the authors of ref.~\cite{sippel1986symmetry} did not find any 4D groups due to a mistake in their treatment of the DEs governing KVFs. These solutions were found by investigating the permitted linear isotropy groups of TppW spacetimes and hence demonstrates the relevance of the CK algorithm  beyond simply classifying spacetimes.

As an immediate application of this work, since the explicit form of the TppW spacetimes is given, these solutions could be perturbed externally to study the stability of these solutions in TEGR in an analogous manner to what has been done in Minkowski space \cite{Landry:2023egz}. As a more ambitious application, we note that there are other GW solutions in GR than ppW spacetimes \cite{Weinberg,Einstein:1937qu,Flanagan:2005yc}. These GW solutions are considered as weak perturbations to Minkowski spacetime when they are propagating in vacuum. However, these GWs do not necessarily propagate in a parallel manner as the TppW solutions. Investigating the teleparallel analogue of these more general GWs, using the methods developed in this paper would be beneficial to the study of TPG theories. 

Alternatively, in GR it is known that GWs could propagate in a (anti-) de Sitter spacetime \cite{podolsky1998interpretation,bishop2016gravitational}. The teleparallel analogue of deSitter spacetimes and anti-deSitter spacetimes have already been found \cite{Coley:2023dbg,mcnutt2024locally}. It is reasonable to assume that GW solutions such as those in ref. \cite{podolsky1998interpretation} could be found hidden in the CSI teleparalllel geometries.

 \begin{longtable}[c]{| c | c | c | c |  c |}

 \hline
 \multicolumn{5}{| c |}{Comparison of Cases}\\
 \hline
 S $\&$ P Cases & $G_d$ &  $\tilde{G}_d$ & Orbit & Conditions  \\
 \hline

 \endfirsthead

 \endhead


 \hline\hline
 \endlastfoot

 Case 1 & 1 & 1 & 1 & none \\
 Case 2 & 2 & 2 & 2 & \eqref{eqn:Case2}  \\
 Case 3 & 2& 2 & 2 & \eqref{eqn:Case3}   \\
 Case 4 & 2 & 2 & 2 & \eqref{eqn:Case4}   \\
 Case 5 & 3 & 3 & 3 & \eqref{eqn:Case5}   \\
 Case 6 & 3 & 3 & 3 & \eqref{eqn:Case6}   \\
 Case 7 & 3 & 3 & 3 & \eqref{eqn:Case7}   \\
 Case 8 & 3 & 3 & 3 & $\epsilon \neq 0$, \eqref{eqn:Case8b}   \\
 Case 8 & 3 & 3 & 3 & $\epsilon = 0$, \eqref{eqn:Case8a}    \\
 n/a & 3 & 3 & 3 & \eqref{eqn:CasemissingG3}   \\
  n/a & 4 & 4 & 4 & \eqref{eqn:CaseMissingG4} \\
 Case 9 & 5 & 5 & 5 & \eqref{eqn:case9}   \\
 Case 10 & 5 & 3 & 3 & \eqref{G5H0}, \eqref{G5M0}\\
 Case 11 & 6 & 4 & 4 & \eqref{eqn:Case11a}, \eqref{eqn:Case11b} \\
 Case 12 & 6 & 4 & 4 & \eqref{G5H0}, \eqref{eqn:Case12a}, \eqref{eqn:Case12b} \\
 Case 13 & 6 & 4 & 4 & \eqref{eqn:Case13a}, \eqref{eqn:Case13b} \\
 Case 14 & 6 & 4 & 4 & \eqref{G5H0}, \eqref{eqn:Case14a}, \eqref{eqn:Case14b} \\
 Case 15 & 6 & 4 & 3 & \eqref{eqn:Case15a}, \eqref{eqn:Case15b} \\
 Case 16 & 7 & 5 & 4 & \eqref{eqn:Case16} \\ 
 Case 17 & 7 & 5 & 4 & \eqref{eqn:Case17} \\

 \end{longtable}

\section*{Acknowledgments}

AL is supported by an Atlantic Association of Research in Mathematical Sciences (AARMS) fellowship. DDM was supported by the Norwegian Financial Mechanism 2014-2021 (project registration number 2019/34/H/ST1/00636). 

\appendix

\section{Components of the torsion tensor in the proper frame}\label{sec:appendixA}

The vector and axial parts of the torsion are, respectively:

\beq
V_a &= \sqrt{2}(Re(M)_{,x} - i Im(M)_{,y}), \\
A_a &= \frac{\sqrt{2}}{3} (Im(M)_{,x}- i Re(M)_{,y}).
\eeq

For the tensor part of the torsion tensor, we will only display the algebraically independent components relative to the complex null coframe basis. All other components can be recovered from index symmetries of the torsion tensor or by complex conjugation.

\begin{align}
\hat{T}_{212} =& \frac{\sqrt{2}}{3} ( Re(M)_{,x} + i Im(M)_{,y}),\\
\hat{T}_{223} =& \sqrt{2} (\bar{M} (\bar{M}_{,x} + i\bar{M}_{,y}) + M(\bar{M}_{,x}-i\bar{M}_{,y})-\bar{M}_{,u} + H_{,x}-iH_{,y}),\\
\hat{T}_{234} =& \frac{\sqrt{2}}{3}(Im(M)_{,x} +iRe(M)_{,y}), \\
\hat{T}_{323} =& -\sqrt{2}(\bar{M}_{,x}-i\bar{M}_{,y}).
\end{align}


\bibliographystyle{apsrev4-2}
\bibliography{TPppwave-Reference-file}

\newpage

\appendix

\end{document}